\documentclass[runningheads]{llncs}
\usepackage{mathrsfs}
\usepackage{graphicx}
\usepackage{listings}
\usepackage{amssymb}
\usepackage{amsmath}
\usepackage{mathpartir}
\usepackage[inference]{semantic}
\usepackage{stackengine}
\stackMath
\usepackage{ifthen}
\usepackage{paralist}
\usepackage{xspace}
\usepackage{tikz}
\usetikzlibrary{arrows,automata,positioning,calc}

\usepackage[textsize=footnotesize,shadow]{todonotes}
\usepackage{upgreek}
\usepackage{apxproof}

\newcounter{sqindex}

\newtheoremrep{theorem}{Theorem}
\newtheoremrep{proposition}{Proposition}

\makeatletter 
\def\arcr{\@arraycr}
\makeatother

\newcommand{\xRightarrow}[2]{\underset{#2}{\overset{#1}{\Rightarrow}}}

\makeatletter
\newcommand{\xleftrightarrow}[2][]{\ext@arrow 3359\leftrightarrowfill@{#1}{#2}}
\newcommand{\xdashrightarrow}[2][]{\ext@arrow 0359\rightarrowfill@@{#1}{#2}}
\newcommand{\xdashleftarrow}[2][]{\ext@arrow 3095\leftarrowfill@@{#1}{#2}}
\newcommand{\xdashleftrightarrow}[2][]{\ext@arrow 3359\leftrightarrowfill@@{#1}{#2}}
\def\rightarrowfill@@{\arrowfill@@\relax\relbar\rightarrow}
\def\leftarrowfill@@{\arrowfill@@\leftarrow\relbar\relax}
\def\leftrightarrowfill@@{\arrowfill@@\leftarrow\relbar\rightarrow}
\def\arrowfill@@#1#2#3#4{%
	$\m@th\thickmuskip0mu\medmuskip\thickmuskip\thinmuskip\thickmuskip
	\relax#4#1
	\xleaders\hbox{$#4#2$}\hfill
	#3$%
}
\makeatother

\newcommand\mydef{\mathrel{\overset{\makebox[0pt]{\mbox{\normalfont\tiny\sffamily def}}}{=}}}
\newcommand\ttt{\textbf{tt}}
\newcommand\fff{\textbf{ff}}

\def\plusplus{{\ \hspace{-.05em}\raisebox{.4ex}{\tiny\bf ++}}}

\newcommand{\prop}{{\cal P}\xspace}
\newcommand{\FlagDate}{{Flagging Monitor}\xspace}
\newcommand{\Flagdate}{{Flagging monitor}\xspace}
\newcommand{\flagdate}{{flagging monitor}\xspace}
\newcommand{\Fdate}{{Monitor}\xspace}
\newcommand{\fdate}{{monitor}\xspace}
\newcommand{\Fdates}{{{Monitor}s}\xspace}
\newcommand{\fdates}{{{monitor}s}\xspace}
\renewcommand{\date}{{DATE}\xspace}

\begin{document}
\title{Incorporating Monitors in Reactive Synthesis without Paying the Price\thanks{This research is funded by the ERC consolidator grant D-SynMA under the European Union’s Horizon 2020 research and innovation programme (grant agreement No 772459).}}

\author{Shaun Azzopardi\orcidID{0000-0002-2165-3698} \and
Nir Piterman\orcidID{0000-0002-8242-5357} \and
Gerardo Schneider\orcidID{0000-0003-0629-6853}}
%
\authorrunning{S. Azzopardi et al.}
%
\institute{University of Gothenburg, Gothenburg, Sweden\\
\email{\{name.surname\}@gu.se}
}
\maketitle              
\begin{abstract}
  Temporal synthesis attempts to construct reactive programs that
satisfy a given declarative (LTL) formula. Practitioners have found it challenging to work exclusively with declarative specifications, and have found languages that combine modelling with declarative specifications more useful. Synthesised controllers may also need to work with pre-existing or manually constructed programs.
In this paper we explore an approach that combines synthesis of declarative specifications in the presence of an existing behaviour model as a monitor, with the benefit of not having to reason about the state space of the monitor.
We suggest a formal language with automata monitors as non-repeating and repeating triggers for LTL formulas.
We use symbolic automata with memory as triggers,
resulting in a strictly more expressive and succinct language than existing regular expression triggers.
We give a compositional synthesis procedure for this language, where reasoning about the monitor state space is minimal.
To show the advantages of our approach we apply it to specifications requiring counting and constraints over arbitrarily long sequence of events, where we can also see the power of parametrisation, easily handled in our approach. 
We provide a tool to construct controllers (in the form of symbolic automata) for our language.

  \keywords{synthesis  \and temporal logic \and symbolic automata \and monitoring}
\end{abstract}

\section{Introduction}

Synthesis of programs from declarative specifications is an attractive
prospect. Although thought prohibitive due to the theoretical hardness
of LTL synthesis, recent improvements have made it a more reasonable
endeavour, e.g. the identification of GR(1)
\cite{10.1007/11609773_24}, for which synthesis is easier, and
development of tools such as Strix \cite{strix,strix1} whose
decomposition method allows for practical synthesis of full
LTL.
Limitations remain in the context of LTL, due to the inherent hardness
of the problem.
Beyond LTL there are also directions where the practicality of
synthesis is not clear.

In addition to these algorithmic challenges, there are additional
methodological challenges.
Practitioners have identified that it is sometimes very challenging to
write declarative specifications, and suggested to use additional
modelling \cite{DBLP:journals/corr/FilippidisMH16,maoz2019spectra}.
Furthermore, synthesised parts need to work alongside pre-existing or
manually constructed parts (cf. \cite{DBLP:conf/fossacs/LustigV09}).
This, however, further exacerbates the algorithmic challenge as the
state-space of the additional parts needs to be reasoned about by the
synthesis algorithm.

We argue that modelling could also be a practical way of dealing with
some of the algorithmic challenges and advocate a partial use of synthesis, leaving parts that are
impractical for synthesis to be manually modelled. This leaves the question of how to combine the two parts. 

We suggest to compose automata with synthesised controllers by transfer of control rather than co-operation. 
We define a specification language with repeating and non-repeating \textit{trigger}
properties (cf. \cite{forspec,10.1007/978-3-642-17511-4_18}).
Triggers are defined as   
environment observing automata/monitors, which transfer control to LTL formulas.
Both aspects~--~control transfer and triggers~--~are familiar to practitioners and would be easy to use:
control transfer is natural for software; and 
triggers are heavily used in industrial verification languages (cf. \cite{forspec}). 

We aim at triggers that are rich, succinct and easy to write. 
Thus, we use monitors extracted from symbolic executable automata inspired by DATEs \cite{10.1007/978-3-642-03240-0_13}.
Expressiveness of automata is increased by having variables that are updated by guarded transitions, which means that auomata can be infinite-state (but the benefits remain if they are restricted to finite-state).
This choice of monitors allows to push multiple other interesting
concerns that are difficult for LTL synthesis to the monitor side. 
Experience of using such monitors in the runtime verification community suggests that they are indeed easy to write \cite{DBLP:conf/rv/FalconeKRT18}.

Our contributions are as follows.
We formally define our specification language ``monitor-triggered temporal logic''.
We show that the way we combine monitors with LTL indeed bypasses the need to reason about the state-space of monitors.
Thus, avoiding some of the algorithmic challenges of synthesis.
We briefly present our synthesis tool.
We give examples highlighting the benefits of using monitors, focusing on counting (with appropriate counter variables updated by monitor transitions) and parametrisation (with unspecified parameter variables that can be instantiated to any required value). Full proofs of the propositions and theorems claimed can be found in the appendix.\\

\noindent
\textbf{Related Work}
In the literature we find several approaches that use monitors in the
context of synthesis. Ulus and Belta use monitors with reactive
control for robotic system navigation, with monitors used for
lower-level control (e.g. to identify the next goal locations), and
controllers used for high-level control to avoid conflicts between
different robots \cite{DBLP:conf/rv/UlusB19}.
Wenchao et al.~consider
human-in-the-loop systems, where occasionally the input of a human is
required. The controller monitors the environment for any possible
violations, and invokes the human operator when necessary
\cite{10.1007/978-3-642-54862-8_40}.\footnote{
We can think of our approach as dual, where the monitor invokes the
synthesised controller when necessary.
}
The use of monitors in these approaches is {\em ad hoc}, a more general
approach is that of the \texttt{Spectra} language
\cite{maoz2019spectra}. Essentially, Spectra monitors have an initial
state, and several safety transition rules of the form $p \rightarrow
q$, where $p$ is a proposition on some low-level variables, and $q$
defines the \texttt{next} value of the monitor variable.
This monitor variable can be used in the higher-level controller
specification.
The approach here is more general than ours in a sense, since we limit
ourselves to using monitors as triggers, however our monitors are more
succinct and expressive. 

The notion of triggers in temporal logic is not new, with regular
expressions being used as triggers for LTL formulas in different languages
\cite{forspec,390251,10.5555/2540128.2540252,DBLP:journals/iandc/FaymonvilleZ17}.
Complexity wise, Kupferman et al. show how the synthesis of these trigger properties is 2EXPTIME-complete \cite{10.1007/978-3-642-17511-4_18}.
However, in order to support such logics algorithms would have to
incorporate the entire state-space of the automata induced by the regular
expression triggers. We are not aware of implementations supporting synthesis from such
extensions of LTL.
Using automata directly within the language, as we do, may be more
succinct and convenient.
We also include a \emph{repetition} of trigger formulas in a way that
is different from these extensions.
However, the main difference is in avoiding the need to reason about
the triggering parts. 

Our combination of monitors and LTL formulas can be seen as a control-flow composition \cite{DBLP:conf/fossacs/LustigV09}. 
Lustig and Vardi discuss how to synthesise a control-flow composition that satisfies an LTL formula given an existing library of components. They consider all components to be given and synthesise the composition itself.
Differently, we assume the composition to be given and synthesise a controller for the LTL part. 
Other work given a global specification reduces it according to that of the existing components, resulting in a specification for the required missing component \cite{DBLP:journals/entcs/Raclet08}. This is at a higher level than our work, since we start with specifications for each component.

\nocite{DBLP:journals/fmsd/KupfermanV01}
\nocite{DBLP:reference/mc/Kupferman18}

\section{Preliminaries}


We write $\sigma$ for infinite traces over an event alphabet $\Sigma$. We use the notation $\sigma_{i,j}$,
  where $i, j \in \mathbb{N}$ and $i\leq j$, to refer to the sub-trace of $\sigma$
  starting from position $i$, ending at (including) position $j$. We
  write $\sigma_i$ for $\sigma_{i,i}$, and $\sigma_{i,\infty}$ for the
  suffix of $\sigma$ starting at $i$. \smallskip

\noindent
\textbf{Linear Temporal Logic (LTL)} General LTL ($\phi$) and co-safety LTL ($\varphi$) are defined
  over a set of propositions $\prop$ respectively as follows, where
  $e\in \prop$:
  \begin{align*}
    \phi &\mydef \ttt \mid \fff \mid e \mid \neg e \mid \phi \wedge \phi \mid \phi \vee \phi \mid X\phi \mid \phi U \phi \mid G\phi\\
    \varphi &\mydef \ttt \mid \fff \mid e \mid \neg e \mid \varphi \wedge \varphi \mid \varphi \vee \varphi \mid X\varphi \mid \varphi U \varphi
  \end{align*}

\noindent We also define and use $F \phi \mydef {\ttt}U \phi$ and $\phi W \phi' \mydef (\phi U \phi') \vee G \phi$.
We write 
$\sigma \vdash \phi$ for $\sigma_{0,\infty}\vdash \phi$.
We omit the standard semantics of LTL \cite{DBLP:reference/mc/PitermanP18}.\footnote{This is also available in the Appendix~A.}\smallskip


\noindent
\textbf{Mealy Machines}  A \emph{Mealy machine} is a tuple $C = \langle S, s_0, \Sigma_{in},
  \Sigma_{out}, \rightarrow,F\rangle$, where $S$ is the set of states,
  $s_0$ the initial state, $\Sigma_{in}$ the set of input events,
  $\Sigma_{out}$ the set of output events, $\rightarrow : S \times
  \Sigma_{in} \mapsto \Sigma_{out} \times S$ the complete
  deterministic transition function, and $F\subseteq S$ a set of
  accepting states. For $(s,I,O,s') \in \rightarrow$
  we write $s \xrightarrow{I/O} s'$. 
  

Notice that by definition for every state $s\in S$ and every $I\in
\Sigma_{in}$ there is $O\in \Sigma_{out}$ and $s'$ such that
$s \xrightarrow{I/O} s'$. 
A \emph{run} of the Mealy machine $C$ is $r=s_0,s_1,\ldots$ such that
for every $i\geq 0$ we have $s_i \xrightarrow{I_i/O_i} s_{i+1}$ for some
$I_i$ and $O_i$. A run $r$ \emph{produces} the word
$w=\sigma_0,\sigma_1,\ldots$, where $\sigma_i = I_i\cup O_i$.
We say that $C$ produces the word $w$ if there exists a run $r$
producing $w$.
We say that $C$ \emph{accepts} a prefix $u$ of $w$ if $s_{|u|}\in F$. \smallskip

\noindent
\textbf{Realisability}  An LTL formula $\varphi$ over set of events $\prop = \prop_{in}
  \cup \prop_{out}$ is \emph{realisable} if there exists a
  Mealy machine $C$ over input events $2^{\prop_{in}}$ and output events $2^{\prop_{out}}$
  such that for all words $w$ produced by $C$ we
  have $w\vdash \varphi$. We say $C$ realises $\varphi$. 

\begin{theorem}[{\rm \cite{DBLP:conf/popl/PnueliR89}}]
  Given an LTL formula $\varphi$ it is decidable in 2EXPTIME whether
  $\varphi$ is realisable.
  If $\varphi$ is realisable the same algorithm can be used to
  construct a Mealy machine $C_\varphi$ realising $\varphi$.
\end{theorem}


\subsection{{\FlagDate}s}

We introduce our own simplified version of {\date}s
\cite{10.1007/978-3-642-03240-0_13,Shaun_Azzopardi48593129}, \textit{\flagdate}s, as a formalism for defining runtime monitors.
{\Flagdate}s (\fdates, for short) are different from {\date}s in that they work
in discrete time, and events are in the form of sets.
\Fdates are designed such that once they \emph{flag} (accept) they
never flag again.
This is modeled by having
\emph{flagging} states, which are used to signal that monitoring has
ended successfully. We also use \emph{sink}
states, from which it is assured the \fdate cannot flag in the
future.
We ensure that the \fdate flags only upon determining a matching sub-trace, and thus a \fdate upon reaching a flagging
state can never flag again. \smallskip

\noindent
\textbf{Monitor}
  A \emph{\fdate} is a
  tuple $D = \langle \Sigma, \mathbb{V}, \Theta, Q,
  \theta_0, q_0, F, \perp, \rightarrow\rangle$, where $\Sigma$ is the event alphabet, $\mathbb{V}$ is a set of typed variables, $\Theta$ is the set of possible valuations of $\mathbb{V}$, $Q$ is a finite set of states, $\theta_0 \in \Theta$ is the initial variable valuation, $q_0 \in Q$ is the initial state, $F \subseteq (Q \setminus \{q_0\})$ is the set of \emph{flagging states} (we often use $q_F \in F$), $\perp \in Q$ is a \emph{sink} state, and $\rightarrow \in Q_\top \times (\Sigma \times \Theta
    \mapsto \{\textit{true}, \textit{false}\}) \times
    (\Sigma \times \Theta \mapsto \Theta) \mapsto Q$
    is the \emph{deterministic transition function}, from $Q_\top \mydef Q \setminus \{\perp\}$, activated if a
    guard holds on the input event and the current variable
    valuation, while it may perform some action to transform the
    valuation.

  For $(q,g,a,q') \in \rightarrow$ we
  write $q \xrightarrow{g \mapsto a} q'$, and we will be using $E$ as the input event parameter for both $g$ and $a$. We omit $g$
  when it is the \textit{true} guard,
  and $a$ when it is the \textit{null} action.
  We use $D_{\langle * \rangle}$ for the \fdate that accepts
  on every event, i.e. $\langle \Sigma, V, \Theta, \{q_0, q_F,\perp\},
  \theta_0, q_0, \{q_F\}, \perp, \{q_0
  \xrightarrow{ true \mapsto null} q_F\} \rangle$.

For example, the \fdate in Figure~\ref{fig:countingdate} keeps a counter that counts the number of \textit{knock} events, and flags when the number of knocks is exactly $n$.


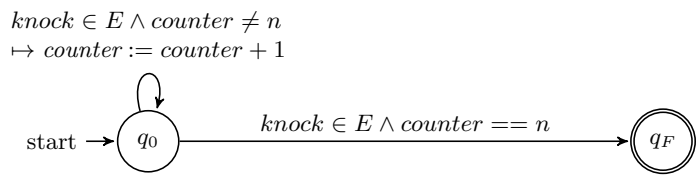
\begin{figure}[tbp!]\centering
	\begin{tikzpicture}[->,>=stealth',shorten >=1pt,auto,node distance=2.8cm,semithick,cross/.style={path picture={\draw[black](path picture bounding box.south east) -- (path picture bounding box.north west) (path picture bounding box.south west) -- (path picture bounding box.north east);}}]
		\tikzstyle{every state}=[text=black]
		
		\node[initial,state] (A)                    {$q_0$};
		\node[state, accepting]         (acc) [right = 6cm of A] {$q_F$};
		
		\path (A) edge [loop above]              node {$\begin{array}{l}
				knock \in E \wedge counter \neq n\\ \mapsto \textit{counter} := counter + 1
			\end{array}$} (A)
		edge []             node {$knock \in E \wedge counter == n$} (acc);
	\end{tikzpicture}
	\caption{\Fdate that counts the number of knocks, and flags after $n$ knocks.}
	\label{fig:countingdate}
\end{figure}

We give an operational semantics to \textit{\fdate}s, with
configurations as pairs of states and valuations, with transitions
between configurations tagged by events. \smallskip

\noindent
\textbf{Monitor Semantics}
\label{def:fdatesemantics}
The semantics of {\flagdate}s \cite{Shaun_Azzopardi48593129}
is given over configurations of type $Q \times \Theta$, with
transitions labeled by $\Sigma$, and the
transition $\rightarrow$ defined by the following rules:
\begin{inparaenum}[(1)]
  \item
    A transition from a non-flagging and non-sink configuration is
    taken when the guard holds on the event and valuation, and then the
    latter is updated according to the transition's action;
  \item
    If there is no available transition whose guard holds true on the 
    current valuation then transition to the same configuration
    (stutter);
  \item
    A sink configuration cannot be left; and
  \item
    A flag configuration always transitions to the sink configuration.
\end{inparaenum}
We use $\Rightarrow$ for the transitive closure of $\rightarrow$.\footnote{See Definition~\ref{def:fullmonsemantics} in Appendix~A for full formal semantics.}\smallskip 

\begin{toappendix}
	\section{Full Definitions and Proofs for Section 2 (Background)}

	\begin{definition}[General LTL Satisfaction]
		\[
		\begin{array}{lll}
			\sigma_{i,j} \vdash \ttt &\mydef& \textit{true}\\
			\sigma_{i,j} \vdash \fff &\mydef& \textit{false}\\
			\sigma_{i,j} \vdash e &\mydef& e \in \sigma_{i,i}\\
			\sigma_{i,j} \vdash \neg\phi &\mydef& \neg (\sigma_{i,j} \vdash \phi)\\
			\sigma_{i,j} \vdash \phi \wedge \phi' &\mydef& (\sigma_{i,j} \vdash \phi) \wedge (\sigma_{i,j} \vdash \phi')\\
			\sigma_{i,j} \vdash \phi \vee \phi' &\mydef& (\sigma_{i,j} \vdash \phi) \vee (\sigma_{i,j} \vdash \phi')\\
			\sigma_{i,j} \vdash X\phi &\mydef& j > i \wedge \sigma_{i+1, j} \vdash \phi\\
			\sigma_{i,j} \vdash \phi U \phi' &\mydef& \exists l \cdot l \leq j \wedge \sigma_{l,\infty} \vdash \phi' \wedge \forall k \cdot i \leq k < l \wedge \sigma_{k,j} \vdash \phi\\ 
			\sigma_{i,j} \vdash G\phi &\mydef& j = \infty \wedge \forall k \cdot k \geq i \implies \sigma_{k,j} \vdash \phi
		\end{array}
		\]
	\end{definition}

\begin{definition}\label{def:fullmonsemantics}
		The semantics of {\flagdate}s \cite{Shaun_Azzopardi48593129}
	is given over configurations of type $Q \times \Theta$, with
	transitions labeled by $\Sigma$, and the
	transition $\rightarrow$ defined by the following rules:

		\begin{enumerate}
		\item A transition from a non-flagging and non-sink configuration is
		taken when the guard holds on the event and valuation, and the
		latter updated according to the transition's action:
		\begin{mathpar}
			\inferrule[]
			{q_1 \xrightarrow{ g \mapsto a} q_2 \qquad g(E, \theta) \qquad q_1 \not\in F \cup \{\perp\}}
			{(q_1,\theta) \xrightarrow{E} (q_2, a(E, \theta))}
		\end{mathpar}
	
		\item If there is no available transition whose guard holds true on the 
		current valuation then transition to the same configuration
		(stutter):
		\begin{mathpar}
			\inferrule[]
			{\nexists q_1 \xrightarrow{ g \mapsto a} q_2 \cdot g(E, \theta)}
			{(q_1,\theta) \xrightarrow{E} (q_1,\theta)}
		\end{mathpar}
	
		\item A sink configuration cannot be left:
		\begin{mathpar}
			\inferrule[]
			{\quad}
			{(\perp,\theta) \xrightarrow{E} (\perp, \theta)}
		\end{mathpar}
	
		\item A flag configuration always transitions to a sink configuration.
		\begin{mathpar}
			\inferrule[]
			{q_F \in F}
			{(q_F,\theta) \xrightarrow{E} (\perp, \theta)}
		\end{mathpar}
	\end{enumerate}
\end{definition}

\end{toappendix}


\noindent
\textbf{Flagging Trace} A finite trace is said to be flagging if it reaches a flagging state.
	$
		\sigma_{i,j} \Vdash D \mydef \exists q_F, \theta' \cdot (q_0, \theta_0) \xRightarrow{\sigma_{i,j}}{} (q_F,\theta').
	$

The semantics ensures that every extension of a flagging trace is non-flagging.

\begin{propositionrep}
	\label{prop:fdatenoextensions}
  $\forall \sigma\in\Sigma^\omega \cdot \forall n \in \mathbb{N} \cdot \sigma_{i,j} \Vdash D \wedge n > 0 \implies \sigma_{i,j + n} \not\Vdash D$.
\end{propositionrep}
\begin{appendixproof} 
	If $\sigma_{i,j} \Vdash D$, then $(q_0, \theta_0) \xRightarrow{\sigma_{i,j}}{} (q_F,\theta')$. By the last rule of the semantics Def.~\ref{def:fdatesemantics} then $(q_0, \theta_0) \xRightarrow{\sigma_{i,j + n}}{} (\perp,\theta')$, and by the third rule of the semantics $(\perp,\theta')$ is a sink state, and thus $q_F$ cannot be reached. Thus no extension can satisfy $D$, since satisfaction requires exactly reaching $q_F$.
\end{appendixproof}

We can also easily show that $D_{\langle * \rangle}$ accepts all traces of length one.

\begin{propositionrep}
\label{prop:emptydateacceptssingleevent}
$\forall \sigma\in\Sigma^\omega$ and $\forall i \in \mathbb{N} \cdot \sigma_{i,i} \Vdash D_{\langle * \rangle}$.
\end{propositionrep}
\begin{appendixproof} 
	Given a general $i$ then $\sigma_{i}$ is of the form $\langle E \rangle$, and by definition of $D_{\langle * \rangle}$ and by the first rule of Defn.~\ref{def:fdatesemantics} then $(q_0, \theta_0) \xRightarrow{E}{} (q_F, \theta_0)$. By the definition of flagging traces then $\sigma_{i} \Vdash D_{\langle * \rangle}$ follows immediately.
\end{appendixproof}


\section{Monitors as Triggers for LTL Formulas}


We suggest a simple kind of interaction between \fdates and LTL, where
\fdates are used as \emph{triggers} for LTL. Previous work has
considered the use of a trigger operator that activates the
checking of an LTL expression when a certain regular expression matches \cite{forspec}. Our approach here is similar, except that we maintain a stricter separation between the monitored and temporal logic parts. 

Our language combining \fdates with LTL has three operators: (i) \fdates
as a trigger for an LTL formula; (ii) repetition of the
trigger formula (when the LTL formula is co-safety); and (iii) assumptions in
the form of LTL formulas.

\begin{definition}[Monitor-Triggered Temporal Logic]
  Monitor-triggered temporal logic extends LTL with three operators: 
  \[
  \begin{array}{lll}
    \pi' &=& D\text{:}\phi \mid (D;\varphi)^*\\
    \pi &=& \phi \rightarrow \pi'
  \end{array}
  \]

  \noindent
  Formula $D\text{:}\phi$ denotes the triggering of an LTL formula
  $\phi$ by a \fdate $D$. We call formulas of this form \emph{simple-trigger LTL}.
  Formula $(D;\varphi)^*$ repeats infinitely the triggering of a
  co-safety LTL formula $\varphi$ by a \fdate $D$.
  We call formulas of this form \emph{repeating-trigger LTL}.
  Finally, $\phi \rightarrow \pi'$ models a specification with an LTL
  assumption $\phi$. LTL formulas are defined over a set of propositions $\prop = \prop_{in} \cup \prop_{out}$ and {\fdate}s over the alphabet $\Sigma = 2^{\prop_{in}}$.

  %

\end{definition}

In formulas of the form $D{:}\phi$ if $D$ flags then the suffix must satisfy $\phi$.
In formulas of the form $(D;\varphi)^*$, the monitor restarts after
satisfaction of the co-safety formula $\varphi$.
For example, if $D$ is the \fdate in Figure~\ref{fig:countingdate} and
$\varphi = (\textit{open} \wedge X (\textit{greet} \wedge X
\textit{close})))$, then $D{:}\varphi$ would accept every trace that waits  
for knocks, and at the $n$th knock opens the door, then 
greets, and then closes the door. On the other hand, $(D;\varphi)^*$
requires the trace to arbitrarily repeat this behaviour. 

To support repeating triggers, we define the notion of tight satisfaction.

\begin{definition}[Tight Co-Safety LTL Satisfaction]
	\label{def:tightsat}
  A finite trace is said to \emph{tightly} satisfy a co-safety LTL
  formula if it satisfies the formula and no strict prefix satisfies
  the formula: 
  $\sigma_{i,j} \Vdash \varphi \mydef \sigma_{i,j} \vdash \varphi
  \wedge (\forall k \cdot i \leq k < j \implies \sigma_{i,k}
  \not\vdash \varphi)$. We also call such a trace a tight witness for the LTL formula.
\end{definition}

Note that here a tight witness is not necessarily a minimal witness
(in the sense that all of its extensions satisfy the LTL formula). For
example, for every set of propositions $P$, a trace $\langle P \rangle$ is a minimal witness for $X \ttt$
\cite{10.1007/978-3-540-77395-5_11}. However it is not a tight witness
in our sense, since $\langle P \rangle \not \vdash X \ttt$. On the
other hand $\langle P, P \rangle$ is a tight witness since $\langle P,
P \rangle \vdash X \ttt$ and every prefix of it does not satisfy $X \ttt$.

Notice that it would not be simple to just use finite trace semantics for
full LTL
\cite{fisman1,fisman,10.5555/2540128.2540252,countingsemantics,10.1007/978-3-540-77395-5_11}.
Consider for example, the trace $\langle \{a\}\rangle$, which
satisfies $G a$.
It is not clear how to define tight satisfaction in order to start the
monitor again.
For example, $\langle \{a\} \rangle$ can be extended to $\langle
\{a\}, \{a\}\rangle$ and still satisfy $G a$. 
Hence formulas of the form $(D;\varphi)^*$ are restricted to
co-safety LTL, where satisfaction over finite traces is
well-defined and accepted.


We now define the trace semantics of the trigger and repetition
operators.


\begin{definition}[Monitor-Trigger Temporal Logic Semantics]
  \label{def:mttlsem}
  \begin{enumerate}
  \item An infinite trace satisfies a simple-trigger LTL formula if when a prefix of it causes the \textit{\fdate} to flag then the corresponding suffix (including the last element of the prefix) satisfies the LTL formula:
    %
    \begin{equation*}
      \sigma_{i,\infty} \vdash D\text{:}\phi \mydef \exists j \cdot i \leq j \wedge (\sigma_{i,j} \Vdash D \implies \sigma_{j, \infty} \vdash \phi)  \text{\qquad where $i \in \mathbb{N}$}.
    \end{equation*}

  \item A finite trace satisfies one step of a repeating-trigger LTL formula if a prefix of it causes the \textit{\fdate} to flag and the corresponding suffix (including the last element of the prefix) tightly satisfies the co-safety LTL formula:
    %
    \begin{equation*}
      \sigma_{i,k} \vdash D;\varphi \mydef \exists j \cdot i \leq j \leq k \wedge (\sigma_{i,j} \Vdash D \wedge \sigma_{j, k} \Vdash \varphi) \text{\qquad where $i,k \in \mathbb{N}$}.
    \end{equation*}
    %

  \item An infinite trace satisfies a repeating-trigger LTL formula if when a prefix of it matches the
    \textit{\fdate} then the corresponding infinite suffix matches the
    LTL formula:	 
    \begin{equation*}
      \sigma \vdash (D;\varphi)^* \mydef \forall i \cdot \sigma_{0,i} \Vdash D \implies \exists j \cdot j \geq i\ \wedge\ \sigma_{0,j} \vdash D;\varphi\ \wedge\ \sigma_{j + 1, \infty} \vdash (D;\varphi)^*.
    \end{equation*}
    
  \item An infinite trace satisfies a specification $\pi'$ with an assumption $\phi$ when if it satisfies $\phi$ it also satisfies $\pi'$:
    
    \begin{equation*}
      \sigma \vdash \phi \rightarrow \pi' \mydef \sigma \vdash \phi \implies \sigma \vdash \pi'.
    \end{equation*}

  \end{enumerate}
\end{definition}

\noindent



\noindent
An interesting aspect of this semantics is that in a formula $D;\varphi$,
$D$ and $\varphi$ share an event, and the same for $D{:}\phi$. This is a choice we make to allow for
message-passing between the two later on. Here it does not limit us,
since not sharing a time step can be simulated by
adding a further transition with a \textit{true} guard before
flagging, or by simply transforming $\phi$ into $X \phi$.

%
%
%
%

This semantics ensures that given an infinite trace, when a finite
sub-trace satisfies $D;\varphi$, extensions of the sub-trace do not also satisfy it.

\begin{propositionrep}
	$\sigma_{i,j} \vdash D;\varphi \implies \forall k > j \cdot \sigma_{i,k} \not\vdash D;\varphi.$\footnote{Proof can be found in Appendix~B.}
\end{propositionrep}
\begin{appendixproof}
	($\implies$) If $\sigma_{i,j} \vdash D;\varphi$ then $\sigma_{i,j}$ can be divided into $\sigma_{i,j'}$ and $\sigma_{j',j}$ such that $\sigma_{i,j'} \Vdash D$ and $\sigma_{j',j} \Vdash \varphi$ is also true.
	
	An easy corollary of Proposition~\ref{prop:fdatenoextensions} is that there if a trace flags then any strict prefix of it does not flag. Then $j'$ is unique here. Furthermore, Definition~\ref{def:tightsat} ensures that $\sigma_{j',j}$ does not a strict prefix that also satisfies $\varphi$. 
	
	Consider for contradiction that $\exists k > j \cdot \sigma_{i,k} \vdash D;\varphi$. Then we know that $\sigma_{i,j'} \Vdash D$ and $\sigma_{j',k} \Vdash \varphi$ (as determined before the choice of $j'$ here is the only choice). Definition~\ref{def:tightsat} here causes a contradiction, since $\sigma_{j',j} \Vdash \varphi$ implies that $\sigma_{j',j} \vdash \varphi$, but $\sigma_{j',k} \Vdash \varphi$ implies that $\sigma_{j',j}
	\not\vdash \varphi$. \qed
\end{appendixproof}

We can prove that a trace $\sigma$ satisfies an LTL formula $\phi$ iff it also satisfies the formula where $\phi$ is triggered by the empty \textit{\fdate}.

\begin{propositionrep}\label{p:removeemptydate}
	$\sigma \vdash \phi \iff \sigma \vdash D_{\langle * \rangle}{:}\phi.$\footnote{Proof can be found in Appendix~B.}
\end{propositionrep}
\begin{appendixproof}
	($\implies$) By definition of $\sigma_{o,\infty} \vdash D_{\langle * \rangle}{:}\phi$, Definition~\ref{def:mttlsem}, then $\exists j \cdot 0 \leq j \wedge (\sigma_{0,j} \Vdash D_{\langle * \rangle} \implies \sigma_{j, \infty} \vdash \phi)$. Since $D_{\langle * \rangle}$ flags upon only one event, by Proposition~\ref{prop:emptydateacceptssingleevent}, then $\sigma_{0,j} \Vdash D_{\langle * \rangle}$ is true for $j = 0$. But $\sigma_{0, \infty} \vdash \phi)$ is equivalent to $\sigma \vdash \phi$, which we assumed. \qed
	
	($\impliedby$) Suppose $\sigma \vdash D_{\langle * \rangle}{:}\phi$. Then $\sigma_{0} \Vdash D_{\langle * \rangle}$, from Proposition~\ref{prop:emptydateacceptssingleevent}. Then by the first rule of Definition~\ref{def:mttlsem} $\sigma_{0,\infty} \vdash \phi$, which is equivalent to the left-hand side. \qed
\end{appendixproof}


%
%

Moreover, we can show that adding these \fdates as triggers for LTL formulas results in a language that is more powerful than LTL. 

\begin{theorem}
  Our language is strictly more expressive than LTL.
  \label{more expressive}
\end{theorem}
\begin{proof}
  Proposition~\ref{p:removeemptydate} shows that every LTL
  formula $\phi$ can be written in our language as
  $D_{\langle*\rangle};\phi$. LTL cannot express the property that each even time step must have $p$ be true \cite{DBLP:journals/iandc/Wolper83} (regardless of what is true at odd steps). In our language $(D_{\langle * \rangle};p \wedge X \ttt)^*$ specifies that $p$ is true in every even time step, and $(D_{\langle *\rangle};X p)^*$ specifies that $p$ is true in every odd
  time step. \qed 
\end{proof}

\begin{figure}[tbp!]\centering
	\begin{tikzpicture}[->,>=stealth',shorten >=1pt,auto,node distance=2.8cm,semithick,cross/.style={path picture={\draw[black](path picture bounding box.south east) -- (path picture bounding box.north west) (path picture bounding box.south west) -- (path picture bounding box.north east);}}]
		\tikzstyle{every state}=[text=black]
		
		\node[initial,state] (A)                    {$q_0$};
		\node[state, accepting]         (acc) [right = 6cm of A] {$q_F$};
		
		\path (A) edge [loop above, above right]              node {
			$\begin{array}{l}
				e \in E \wedge \textit{eventsNo} \div \textit{stepsNo} > n\\ \mapsto \textit{eventsNo} := eventsNo + 1; \textit{stepsNo} := stepsNo + 1
			\end{array}$} (A);
		\path (A) edge [loop below, below right]              node {
			$\begin{array}{l}
				e \not\in E \wedge \textit{eventsNo} \div \textit{stepsNo} > n\\ \mapsto \textit{stepsNo} := stepsNo + 1
			\end{array}$} (A);
		\path (A) edge []             node {$\textit{eventsNo} \div \textit{stepsNo} \leq n$} (acc);
	\end{tikzpicture}
	\caption{\Fdate that keeps track of the number of time steps, and the number of occurrences of $e$, while flagging is the average occurrence of $e$ goes below $n$.}
	\label{fig:averagingdate}
\end{figure}
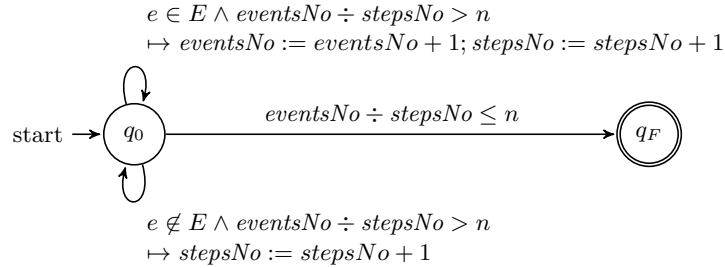

Our logic is even more expressive, for example Figure~\ref{fig:averagingdate} shows a \fdate that flags upon the average occurrence of an event falling below a certain level. We note that, in general, we have not restricted the types of
variables of a monitor to range over finite domains.
Thus, a monitor could also identify context-free or context-sensitive
languages or, indeed, be Turing powerful.
However, Theorem~\ref{more expressive} holds even if we consider only monitors whose variables have finite domains, or even monitors without variables.

\section{Synthesising Monitor-Triggered Controllers}

We have so far discussed our language from a satisfaction 
viewpoint. However we are interested in synthesising systems that
enforce the specifications in our language. In this section we present
our synthesis approach, which relies on the synthesis of controllers
for LTL formulas. 

Consider a specification $\pi = \gamma \rightarrow \pi'$, where $\pi'$ is
either of the form $D{:}\phi$ or $(D;\varphi)^*$.
We focus on specifications where the assumption $\gamma$
is restricted to conjunctions of simple invariants, transition
invariants, and recurrence properties. Formally, we have the
following:
\begin{align*}
  \alpha & \mydef \ttt \mid \fff \mid a \mid \neg \alpha \mid \alpha
  \wedge \alpha \mid \alpha \vee \alpha\\
  \beta & \mydef \alpha \mid X\alpha \mid \beta \wedge \beta \mid
  \beta \vee \beta \\
  \gamma & \mydef G\beta \mid GF \alpha \mid \gamma \wedge \gamma
\end{align*}
That is, $\alpha$ are Boolean combinations of propositions, $\beta$
allows next operators without nesting them, and $\gamma$ is a
conjunction of invariants of Boolean formulas, Boolean formulas that
include next, or recurrence of Boolean formulas.
We discuss below the case of general assumptions.

Let $\pi=\gamma \rightarrow (D{:}\phi)$. Then 
$t(\pi)$ is the formula $\gamma \rightarrow \phi$.
Let $\pi=\gamma \rightarrow (D;\varphi)^*$.
Then $t(\pi)$ is the formula $\gamma \rightarrow \varphi$.
That is, $t(\pi)$ is the specification obtained by considering the
implication of the assumption $\gamma$ 
and the LTL formula.

\subsection{Tight Synthesis for co-Safety Implication Formulas}
Let $\pi$ contain a repeating trigger and let
$t(\pi)=\gamma\rightarrow\varphi$, where $\varphi$ is a co-safety
formula.
Suppose that $t(\pi)$ is realisable and let $C_{t(\pi)}$ be a 
Mealy machine realising $t(\pi)$.

\begin{definition}
A Mealy machine $C$ \emph{tightly realises} a formula of the form
$\gamma\rightarrow \varphi$, where $\varphi$ is a co-safety formula,
if it realises $\gamma\rightarrow \varphi$ and in addition for every word
$w$ produced by $C$ such that $w\vdash \gamma$ there 
exists a prefix $u$ of $w$ such that $C$ accepts $u$,
$u_{0,|u|}\Vdash \varphi$, and for every $u'<u$ we have $C$ does not
accept $u$. 
\end{definition}

That is, when the antecedent $\gamma$ holds, the Mealy machine accepts
the \emph{tight} witness for satisfaction of $\varphi$. 

\begin{theoremrep}\label{thm:tightreal}
  The formula $t(\pi)=\gamma\rightarrow \varphi$ is tightly
  realisable iff it is realisable.
  A Mealy machine tightly realising $t(\pi)$ can be constructed from
  $C_{t(\pi)}$ with the same complexity.
\end{theoremrep}
\begin{proofsketch}
	We construct a deterministic finite automaton that is at most doubly exponential in $\phi$, that accepts all finite prefixes that satisfy $\phi$. Its product with $C_{t(\pi)}$ results in a Mealy machine that accepts all prefixes that satisfy $t(\pi)$, in particular the shortest prefix, as required for realisability.\footnote{Full proof can be found in Appendix~C.}
\end{proofsketch}

\begin{appendixproof}
  Clearly, if $t(\pi)$ is not realisable then it cannot be tightly
  realisable.

  Suppose that $t(\pi)$ is realisable and let $C_{t(\pi)}$ be the
  Mealy machine realising it.
  We show how to augment $C_{t(\pi)}$ with accepting states.

  First, given a co-safety formula $\varphi$, we can construct a
  deterministic finite automaton accepting all \emph{finite} prefixes $u$ such
  that $u_{0,|u|}\vdash \varphi$.
  Construct an alternating weak automaton (AWW) $A_\varphi$ for
  $\varphi$ using standard techniques \cite{DBLP:reference/mc/Kupferman18}.
  The only states of this automaton that have self loops are states of
  the form $\psi_1 U \psi_2$.
  It follows that the only option for this automaton to accept is by a
  transition to $\ttt$.
  We construct a nondeterministic finite automaton (NFW) $N_\varphi$ from
  $A_\varphi$ by using a version of the subset construction and having
  the empty set as the only accepting state of the NFW.
  We then construct a deterministic finite automaton $D_\varphi$ by
  applying the classical subset construction to $N_\varphi$. 

  The AWW $A_\varphi$ is linear in $\varphi$, the NFW $N_\varphi$ is
  at most exponential in $\varphi$, and the DFW $D_\varphi$ is at most
  doubly exponential in $\varphi$. 

  We then take the product of $D_\varphi$ with the 
  Mealy machine $C_{t(\pi)}$.
  
  As $C_{t(\pi)}$ realises $t(\pi)$, whenever a word $w$ satisfies
  $\gamma$ it must satisfy $\varphi$ as well.
  Then, as $D_\varphi$ accepts all prefixes that
  satisfy (tightly) $\varphi$, the first prefix of $w$ that satisfies
  $\varphi$ is accepted by the Mealy machine.%
  \footnote{
  We conjecture that realisability procedures relying on determinisation
  and the Mealy machines constructed from them could be further
  analysed to produce the required tight Mealy machine without
  requiring the additional automaton $D_\varphi$.
  }

  Note that the construction identifying violating prefixes for a
  safety language \cite{DBLP:journals/fmsd/KupfermanV01} would not
  produce the required result:
the automaton constructed using that technique would
  identify the empty trace as satisfying $X \ttt$. \qed
\end{appendixproof}

Note that in the case of tight realisability we can give a controller
with a set of accepting states that enable us to accept upon observing
tight witnesses. In the case where we are only concerned about non-tight realisability we assume the controller does not have any
accepting states. 

\subsection{Monitor-Triggered Synthesis}
We are now ready to handle synthesis for monitor-triggered LTL.

\begin{definition}
	A monitor-triggered LTL formula $\pi$ over set of events $\prop_{in}$
	and $\prop_{out}$ is \emph{realisable} if there exists a
	Mealy machine $C$ over input events $2^{\prop_{in}}$ and output events $2^{\prop_{out}}$
	such that for all words $w$ produced by $C$ we
	have $w\vdash \pi$. We say that $C$ \emph{realises} $\pi$. 
\end{definition}

In the case of simple triggers, we combine the monitor with a Mealy
machine realising $t(\pi)$.
In the case of repeating triggers, we combine the monitor with a
Mealy machine tightly realising $t(\pi)$. In what follows we define the behaviour of the combination of a monitor and a Mealy machine.

Consider a specification $\pi=\gamma \rightarrow \pi'$, where $\pi'$ is
either $M{:}\phi$ or $(M;\varphi)^*$.

\begin{theoremrep}\label{thm:mttlrealis}
	Let $C_{t(\pi)}$ be a Mealy machine realising $t(\pi)$ when $\pi'$ is a simple-trigger LTL, and tightly realising $t(\pi)$ when $\pi'$ is a repeating-trigger LTL.
	Then there is a Mealy machine $M
	\blacktriangleright C_{t(\pi)}$ that realises $\pi$.
\end{theoremrep}
\begin{proofsketch}
	$M
	\blacktriangleright C_{t(\pi)}$ can be constructed over states that correspond to a tuple of $M$ states, valuations, and $C_{t(\pi)}$ states. \Fdate transitions can be unfolded into Mealy machine transitions with no outputs, according to their semantics. Transitions to a flagging state can be composed with transitions from the initial state of $C_{t(\pi)}$. For the repeating case, transitions to final states of $C_{t(\pi)}$ are made instead to point back to the initial configuration (initial state and valuation of $M$, and initial state of $C_{t(\pi)}$). Execution happens only in one machine at a time, except for the shared transition in the repeating case. We show by induction the correctness of this construction.\footnote{Full proof can be found in Appendix~C.}
	
\end{proofsketch}
\begin{appendixproof}
	We give a construction for such a controller, and prove it realises $\pi$.
	
	The \fdate-activated controller $M
	\blacktriangleright C$ is defined as the Mealy machine with the
	following components $\langle S, s_0, \prop_{in},
	\prop_{out},\rightarrow, \emptyset\rangle$, where $S=(Q_M \times \Theta_M) \cup
	Q_C$, $s_0 = (q^0_M,\theta_0)$, $\prop_{in}$ and $\prop_{out}$ are
	shared among all three. The transition relation $\rightarrow$ is the minimal relation respecting: 
	\begin{enumerate}
		
		\item If the \fdate is in a state that has an outgoing transition to a non-final state then take that transition. Note that this also applies to sink states.
		\begin{mathpar}
			\inferrule
			{q'_M \not\in F \qquad (q_M,\theta) \xrightarrow{I}_M (q'_M,\theta')}
			{(q_M,\theta) \xrightarrow[]{I / \emptyset} (q'_M,\theta')}
		\end{mathpar}
		
		\item If the \fdate is in a state that has an outgoing transition to a final state then take that transition synchronously with a corresponding transition from the controller initial state.
		\begin{mathpar}
			\inferrule
			{q'_M \in F_M 
				\qquad
				(q_M,\theta) \xrightarrow{I}_M (q'_M,\theta') \qquad q^0_C \xrightarrow{I/O}_C q_C}
			{(q_M,\theta) \xrightarrow[]{I / O} q_C}
		\end{mathpar}

		\item
		If the controller $C$ is in a non-accepting state take transitions
		from that state. 
		\begin{mathpar}
			\inferrule
			{q_C\notin F_C \qquad q_C\xrightarrow{I/O}_C q'_C}
			{q_C \xrightarrow[]{I / O} q'_C}
		\end{mathpar}
		
		\item
		If the controller $C$ has a transition to an accepting state then transition to the initial state of the \fdate. Note that this will not be activated in the case of simple triggers, for which we will be using controllers without accepting states.
		\begin{mathpar}
			\inferrule
			{q'_C\in F_C \qquad q_C \xrightarrow{I/O} q'_C}
			{q_C \xrightarrow[]{I / O} (q^0_M,\theta^0)}
		\end{mathpar}
		
	\end{enumerate}

To prove this construction realises $\pi$, we consider two cases, when $\pi'$ is a simple-trigger formula $D{:}\phi$ and when it is a repeating-trigger formula $(D{;}\varphi)^*$. Throughout we assume we are considering traces that satisfy the assumption, otherwise the specification is trivially satisfied.

In the case of a simple trigger, it is easy to see that there are two cases: either the monitor never flags, or it flags. If the monitor never flags (i.e. only rule 1 is ever used), then $\pi$ is trivially satisfied.  If the monitor flags, then a suffix of the trace is produced by $C_{t(\pi)}$ (note rules 2 and 3). That is, there is a $j$ such that $\sigma_{0,j} \Vdash D$ and $\sigma_{j, \infty}$ is produced by $C_{t(\pi)}$. Note that rule 4 in the construction is never activated, since $C_{t(\pi)}$ will not have accepting states. However every trace produced by $C_{t(\pi)}$ also satisfies $t(\pi)$, since the former realises the latter 
. The result then easily follows.

In the case of a repeating trigger, there are two cases: either the monitor flags finitely often, or the monitor flags infinitely often. 

\noindent
Suppose the monitor flags finitely often. Consider by induction that the monitor flags zero times, then it is easy to see that the trace satisfies $(D{;}\varphi)^*$, since $D$ never flags. The inductive step, when the monitor flags $n+1$ times, can be easily reduced to $n$th case by pruning from the trace the prefix that satisfies $D{;}\varphi$. This prefix must exist since the monitor must flag at least once and upon entering  $C_{t(\pi)}$, which tightly realises the co-safety formula $t(\pi)$ and thus will accept.

\noindent
Suppose the monitor flags infinitely often, and for contradiction suppose the formula is not realised. There must then be some trace whose prefix does not satisfy $D{;}\varphi$. This prefix in turn must have a prefix that satisfies $D$, otherwise the whole specification $D{;}\varphi$ is satisfied. However, since $D$ must flag on this violating trace then, since $C_{t(\pi)}$ tightly realises the co-safety formula $t(\pi)$, we are guaranteed the prefix will eventually be extended to accept and satisfy $\varphi$. Thus $D{;}\varphi$ must be satisfied by the prefix, which creates a contradiction. \qed
	
\end{appendixproof}

The opposite of Theorem~\ref{thm:mttlrealis} is, however, not
true.
If $\pi$ is realisable then it does not necessarily 
mean that $t(\pi)$ is also realisable.
Consider  
a specification with a \fdate that never
flags, and which thus any Mealy machine realises.
Another example is with a \fdate that
only flags upon seeing the event set $\{a\}$, and an LTL formula of
the form $(b \implies \fff) \wedge (a \implies c)$ (where $a$ and $b$
are input events, and $c$ an output event).
The LTL formula is clearly unrealisable given the first conjunct,
however the combination of the \fdate with a controller for LTL's
second conjunct would realise the corresponding specification.
Thus the construction in Thereom~\ref{thm:mttlrealis} is only sound
but not complete, i.e., we have a procedure to produce controllers for
our language only when the underlying LTL formula (modulo the
assumption) is realisable, or when the monitor cannot flag.

\begin{corollary}
  If $M$ cannot flag, or $t(\pi)$ is realisable, then $\pi$ is realisable. 
\end{corollary}


We recall that we have restricted the assumptions to a combination of
invariants, transition invariants, and recurrence properties.
Such assumptions are ``state-less''.
That is, identifying whether a word satisfies an assumption does not
require to follow the state of the assumption.
Thus, in our synthesis procedure it is enough for the controller to
check whether the assumption holds without worrying about what
happened during the run of the monitor that triggered it.
In particular, if (safety) assumptions are violated only during the
run of monitors, our Mealy machine will still enforce satisfaction of
the implied formula.
In order to treat more general assumptions, we would have to either
analyse the structure of the monitor in order to identify in which
``assumption states'' the controller could be started or give a
precondition for synthesis by requiring that the controller could
start from an arbitrary ``assumption state'', which we leave for future work. Similarly, understanding the conditions the monitor enforces and using them as assumptions would allow us to get closer to completeness of Theorem~\ref{thm:mttlrealis}. One coarse abstraction is simply the disjunction of the monitor's flagging transitions' guards as initial assumptions for the LTL formula.

\section{Tool Support} 

We created a proof-of-concept automated tool\footnote{\url{https://github.com/dSynMa/syMTri}} to support the theory presented in this paper. Implemented in Python, this tool currently accepts as input a \fdate written in a syntax inspired by that of LARVA \cite{CPS09ltr,Shaun_Azzopardi48593129}, and an LTL specification, while it outputs a symbolic representation of the Mealy machine constructed in the proof of Theorem~\ref{thm:mttlrealis}, in the form of a \fdate with outputs.

The proof of Theorem~\ref{thm:tightreal} is constructive and provides an optimal algorithm to synthesise tight controllers using standard automata techniques. For this tool we have instead opted to re-use an existing synthesis tool, Strix \cite{strix,strix1} due to its efficiency. To force Strix to synthesise a tight controller (for repeating triggers), the tool performs a transformation to the co-safety guarantees to output a new event that is only output once a tight witness is detected. This transformation works well on our case studies, but is exponential in the worst-case. This is due to the need for disambiguating disjunctions. For example, given $\psi_1 \vee \psi_2$ we cannot in general easily be sure which disjunct the controller will decide to enforce; instead we disambiguate it to $(\psi_1 \wedge \neg \psi_2) \vee (\neg \psi_1 \wedge \psi_2) \vee (\psi_1 \wedge \psi_2)$ (cf. \cite{DBLP:conf/tacas/BenediktLW13}).

\section{Case Studies}


We have applied our \fdate approach mainly in the setting of conditions on the sequence of environment events, for which synthesis techniques can be particularly inefficient. We will consider a case study involving such conditions, where several events need to be observed before a robot can start cleaning a room. Furthermore we consider a problem from SYNTCOMP 2020 on which all tools timed out due to exponential blowup as the parameter values increase, relating to observing two event buses. We show how our approach using {\fdate}s avoids the pitfalls of existing approaches with regards to these kinds of specifications. 

\subsection{Event Counting}

\begin{figure}[t]\centering
    \begin{tikzpicture}[->,>=stealth',shorten >=1pt,auto,node distance=2.8cm,semithick,cross/.style={path picture={\draw[black](path picture bounding box.south east) -- (path picture bounding box.north west) (path picture bounding box.south west) -- (path picture bounding box.north east);}}]
        \tikzstyle{every state}=[text=black]
        
        \node[initial,state] (A)                    {$q_0$};
        \node[state]         (B) [right =4.2cm of A] {$q_1$};
        \node[state, accepting]         (C) [right = 4cm of B] {$q_2$};

        \path (A) edge [loop above]              node {$\begin{array}{l}
                \textit{inUse} \in E \wedge \textit{inUse} < n\\ \mapsto \textit{inUseFor} \plusplus
            \end{array}$} (A);
        \path (A) edge []              node {$\begin{array}{l}
                \textit{inUse} \in E \wedge \textit{inUseFor} \geq n
            \end{array}$} (B);
        
        \path (B) edge [in=40,out=70,loop, above right]              node {$\begin{array}{l}
                \textit{inUse} \not\in E \wedge unused < m\\ \mapsto \textit{unused} \plusplus
            \end{array}$} (B);
        
        \path (B) edge [out=140,in=110,loop, above left]              node {$\begin{array}{l}
                \textit{inUse} \in E\\ \mapsto \textit{unused} = 0
            \end{array}$} (B);
        
        \path (B) edge []              node {$\begin{array}{l}
                \textit{inUse} \not\in E \wedge \textit{unused} \geq m
            \end{array}$} (C);
    \end{tikzpicture}
    \caption{Flagging monitor that checks that the room has been in use for $n$ time steps, after which when there is a period of $m$ time steps where the room is empty it flags.}
    \label{fig:monitorroomuse}
\end{figure}
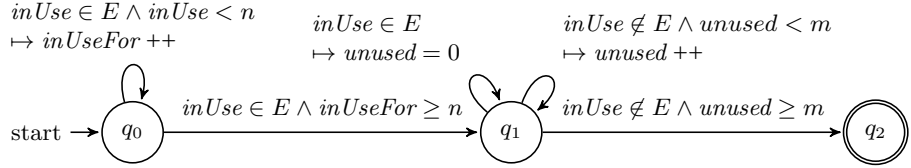

\begin{figure}[htbp!]
	\includegraphics[width=\textwidth]{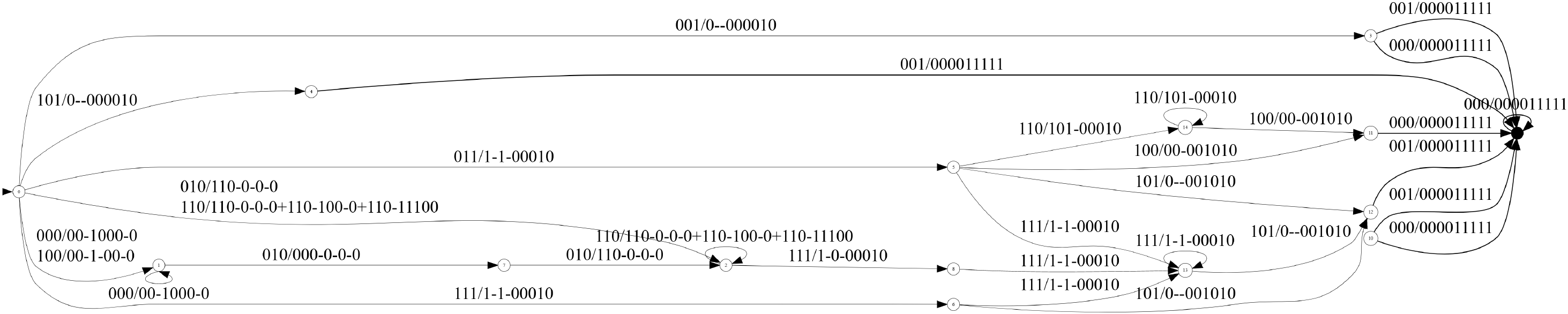}
	\caption{Tight controller for cleaning robot, with rightmost state as accepting state. ($1$ ($0$) in position $i$ means event $i$ (not) occurs, and ${\texttt{-}}$ when we do not care).}
	\label{fig:finalcontroller}
\end{figure}

Consider a break room that is used by people intermittently during the day, and that needs to be cleaned periodically by a cleaning robot. We do not want to activate the robot every time the room is unclean to not disturb people on their break. Instead our procedure involves checking that the room is in use for a certain amount $n$ of time steps. We also do not want the cleaning robot to be too eager or to activate immediately upon an empty room. Thus we further want to constrain the robot's activation on the room being empty for a number of $m$ time steps and reset the counting whenever the room is not empty. We can represent these conditions using the \fdate in Figure~\ref{fig:monitorroomuse}. 

Given a set of assumptions on the environment (e.g. cleaning an unclean locked room eventually results in a clean room), we wish the controller to satisfy that eventually the room will be
clean, after which the robot leaves the room and opens the door to the
public: $F (\textit{isClean}\ \&\ (X F\ !\textit{inRoom})\ \&\ (X F\ !\textit{doorLocked}))$.\footnote{The full specification is available with our tool.} Our tool synthesises Figure~\ref{fig:finalcontroller} as a tight controller for this.



Representing the first monitor condition in LTL is not difficult
($\neg p W (p \wedge X (\neg p W p \wedge ... ))$), where
proposition $p$ corresponds to $\textit{inRoom}$ and $W$ is the weak until operator. 
The second condition is different, given the possible resetting of the
count, but still easily 
representable in LTL
($(\bigvee_{i=0}^{m-1} X^i p) W (\bigwedge_{i=0}^{m-1} X^i\neg p)$).
Setting $n, m = 2$, and $\varphi$ to be what we require out
of the cleaning robot in one step, then a step of our specification (\textit{without repetition}) in LTL is:
$$ \psi = \neg p W (p \wedge X (\neg p W (p \wedge X ((p \vee X p) W (\neg p \wedge
X (\neg p \wedge \varphi)))))).$$

In fact Strix confirms this to be realisable, and
produces an appropriate Mealy Machine with eighty
transitions, the size of which increases with each increase in any of the parameters.

However, using our approach all we require is
Figure~\ref{fig:monitorroomuse} and
Figure~\ref{fig:finalcontroller}. By representing the counting part of
the specification using a \fdate we can create a specification
much more succinct than the LTL one, while its representation
is of the same size for each value of the parameter. Moreover in LTL it is not clear how to reproduce our repeating triggers. 

The difference is that the traditional approaches explicitly enumerate
every possible behaviour and state of the controller at runtime, which
can get very large. In our approach we are instead doing this
symbolically, and allowing the particular behaviour of the environment
at runtime to drive our symbolic monitor. The extra cost associated
with this is the semantics of guard evaluation and maintaining
variable states. For this example, the cost of the variable states
(only two variables) is much smaller than the cost of the Strix
generated machine, while guards simply
check for membership and use basic arithmetic operations.

\subsection{Sequences of Events}

We consider a benchmark from SYNTCOMP 2020 \cite{syntcomp}\footnote{The considered benchmark corresponds to files of the form \texttt{ltl2dba\_beta\_$<$n$>$.tlsf}.}, that generates formulas of the form, e.g. for $n = 2$, $F(p_0 \wedge F(p_1)) \wedge F(q_0 \wedge F(q_1)) \iff G F acc$. Strix \cite{strix}, the best-performing tool in the LTL tracks of the competition, was successful when the bus size was small, however timed out for $n = m = 12$ (and above). The issue here is that the generated strategy must take into account every possible interleaving of the two sequences, which quickly causes a state space explosion.


With our approach we can represent the left-hand side in constant-size for any $n$ and $m$ as illustrated in Figure~\ref{fig:inseq}, where $\textit{maxInSeqX}$ is a function that when $xCount = i$ returns $i$ if $E$ does not contain $x_{i+1}$, and otherwise returns the maximal $j$ such $x_{i+1}, x_{i+2}, ... ,$ and $x_j$ are all in E. The benefits apply however complex the right-hand side.


\begin{figure}[t]
    \begin{tikzpicture}[->,>=stealth',shorten >=1pt,auto,node distance=2.8cm,semithick,cross/.style={path picture={\draw[black](path picture bounding box.south east) -- (path picture bounding box.north west) (path picture bounding box.south west) -- (path picture bounding box.north east);}}]
        \tikzstyle{every state}=[text=black]
        
        \node[initial,state] (A)                    {$q_0$};
        \node[state, accepting]         (acc) [right = 6.5cm of A] {$q_F$};
        
        \path (A) edge [loop above]              node {$
            \begin{array}{ll}
                & \textit{maxInSeqP}(E) \neq n\ \wedge   \textit{maxInSeqQ}(E) \neq n\\  
                & \mapsto \textit{pCount} := \textit{maxInSeqP}(E); \textit{qCount} :=\textit{maxInSeqQ}(E)
            \end{array}$} (A)
        edge []             node {$\textit{maxInSeqP}(E) = n \wedge \textit{maxInSeqQ}(E) = m$} (acc);
    \end{tikzpicture}
    \caption{Event ordering in two buses.}
    \label{fig:inseq}
\end{figure}
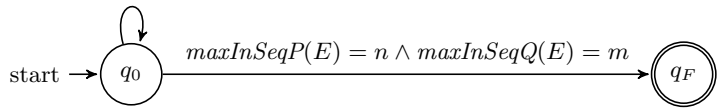

To replicate the whole LTL formula we can use $M{:}GF acc$, where $M$ is the \fdate in Figure~\ref{fig:inseq}. This is somewhat different from the original specification, where a necessary and sufficient relation was specified. One would be tempted to specify this as (for $n = 1$) $(M;(F acc))^*$, however the monitor is not active while the controller is activated, thus $p_1$ and $q_1$ may occur in tandem with $acc$ but be missed by the monitor. Although this is not of consequence towards the satisfaction of the formula here ($p_1$ should occur infinitely often), this is not generally the case. On the other hand $M{:}G F acc$ captures that upon the first activation of $M$ there is no need to monitor the environment's behaviour anymore, and thus is equivalent to the original specification for control.

\section{Discussion}

The case studies we considered in the previous section focused on counting and waiting for sequence of events. We expect other useful applications of \fdates as triggers, given they can be used to specify more sophisticated quantitative properties out of reach for LTL, e.g. see Figure~\ref{fig:averagingdate} \cite{CPS09ltr}.

We have highlighted how our approach extends the scope of use of reactive synthesis. It is clear that we can gain in scalability and expressiveness, but there is a price to pay: the ``trigger'' part. In general, to avoid lack of guarantees one can avoid working directly with automata, and instead use regular expressions or co-safety LTL formulas (under our notion of tight satisfaction) as triggers. Standard inexpensive monitor synthesis \cite{rv-for-tltl} could then be used to generate a \fdate. In the case of more expressive manually-written monitors, which is standard for runtime verification (e.g. \cite{10.1007/978-3-642-03240-0_13}), in practice one can easily apply model checking to the monitor to ensure it satisfies specific properties (e.g. no infinite loops).

There are certain benefits to using a symbolic representation, including succinct representation, and easy parametrisation. The Mealy machine construction we give in the proof of Theorem~\ref{thm:mttlrealis} is in fact not carried out by our tool, but instead it produces a symbolic monitor with outputs that essentially performs the construction on-the-fly. The cost of unfolding is then only paid for the trace at runtime, rather than for all possible traces. Moreover, a symbolic representation allows for specification of parametrised specifications, when parametrisation can be pushed to the monitor side. This can be done by adding any required parameter to the variables of the \fdate, and instantiating its value in the initial variable valuation of the \fdate appropriately. Note that our results are agnostic of the initial valuation, and thus hold regardless of the parameter values.


We have not yet discussed conjunction of trigger formulas, e.g. $(M_1;\psi_1)^* \wedge (M_2;\psi_2)^*$. Conjunction is easy when the output events of $\psi_1$ and $\psi_2$ respectively talk about are independent from each other. Our controller construction can be used independently for each. Similarly, when the properties are safety properties there is no difficulty. However, when, e.g.,  $\psi_1$ is a liveness property with at least one output event correlated with an output event of $\psi_2$, then conjunction is more difficult, due to possible interaction between the two possibly concurrent controllers. We are investigating a solution for this issue of concurrency of controllers by identifying appropriate assumptions about the monitor.

%

Theorem~\ref{more expressive} compares our expressive power to that of
LTL.
We also mention that we do not restrict the variables used by
monitors.
Thus, even when comparing with languages that include
regular expressions or automata
\cite{forspec,10.5555/2540128.2540252,10.1007/978-3-642-17511-4_18} our
language would be more expressive. 
If we were to restrict monitors to be finite state, then, as these
languages can express all $\omega$-regular languages, it is clear that
they would be able to express our specifications. 
We note, however, that the repeating trigger operator is not directly
expressible in these languages.
Thus, the translation involves a conversion of our specification to an
automaton and embedding this automaton in ``their'' specification.
The conversion of our specification to an automaton includes both the
enumeration of the states of the monitor and the exponential
translation of LTL to (tight) automata.

\section{Conclusions}

We have explored synthesis for specifications that combine modelling and declarative aspects, in the form of symbolic monitors triggers for LTL formulas. We have shown how this extends the scope of synthesis by allowing parts of a specification that are hard for synthesis to be instead handled in the monitor part. The synthesis algorithm we give synthesises the LTL part without requiring the need to reason about the monitor.  Moreover, we have implemented this approach and applied it to several case studies involving counting and monitoring multiple sequences of events that can be impossible or hard for LTL synthesis. We showed how by exploiting the symbolic nature of the monitors we can create fixed-size parameterised controllers for some parameterised specifications.  \\


\noindent
\textbf{Future Work}
Our work opens the door to a number of interesting research avenues, both by using richer monitor triggers and by exploring different interactions between triggers and controllers. We discuss below just a few such possibilities. In all the cases below the challenges lie not only in providing a new language to capture the extension but rather in the theoretical framework with a proof that the integration is sound.

A first intuitive extension is to add real-time to the monitors, to express properties like ``\textit{compute the average use of a certain resource every week and activate the controller to act differently depending on whether the average is bigger (or smaller) than a certain amount}''. While extending the monitor with real-time is quite straightforward (our monitors are restricted versions of DATEs \cite{10.1007/978-3-642-03240-0_13} which already contain timers and stopwatches), the challenge will be to combine it with the controller in a suitable manner. Having real-time monitors running in parallel with controllers would enable for instance the possibility to add timeouts to activities performed by the controllers.

Currently we have a strict alternation between the execution of the monitor and the controller: we would like to explore under which conditions the two can instead run in parallel. This would allow the controller to react to the monitor only when certain complex condition hold while the controller is active doing other things (e.g., the monitor might send an interruption request to the controller when a certain sequence of events happens within a certain amount of time, while the controller is busy ensuring a fairness property). 

We could also have many triggers that run in parallel activating different controllers, or even some meta-monitor that acts as an orchestrator to enable and disable controllers depending on certain conditions. This might require to extend/modify the semantics since the interaction might be done asynchronously.

We would like to address the limitation of controller synthesis concerning what to do when the assumptions are not satisfied. It is well-known that in order to be able to automatically synthesise a controller very often one must have strong assumptions, and nothing is said in case the assumptions are not satisfied. We would like to explore the use of \fdates to monitor the violation of assumptions and interact with the controller in order to coordinate how to handle those situations (we can for instance envisage a procedure that automatically extends the controller with transitions that takes the controller to a recovery state if the assumptions are violated).



 \bibliographystyle{splncs04}
 \bibliography{references}
 
\newpage

\appendix

\end{document}